\documentclass[journal]{IEEEtran}

\usepackage[english]{babel}
\usepackage[T1]{fontenc}
\usepackage[utf8]{inputenc}

	\usepackage[pdftex]{graphicx}	
	\DeclareGraphicsExtensions{.pdf,.png,.jpg, .JPG}

\graphicspath{ {./imag/22-10-17/} } 

\usepackage{amsmath,amssymb,amsthm}
\usepackage{amsfonts}
\usepackage{setspace}
\usepackage{cite}
\usepackage{booktabs}
\usepackage{array}

\usepackage{listings}
\usepackage{xcolor}

\usepackage{verbatim} 
\usepackage{float}
\usepackage{bm}
\usepackage{algorithm}
\usepackage[noend]{algpseudocode}
\usepackage{hyperref}

\definecolor{light-gray}{gray}{0.95}
\definecolor{gray}{gray}{0.5}

\lstnewenvironment{listing}[1][]
   {\lstset{#1}\pagebreak[0]}{\pagebreak[0]}
 
\lstdefinestyle{matlab}
   {language=matlab,
   }

\newcommand{\bs}[1]{\boldsymbol{#1}}

\newcommand{\etal}{\emph{et al.}}

\newcommand{\ie}{\emph{i.e.}}

\newtheorem{assumption}{Assumption}

\newtheorem{proposition}{Propositon}

\begin{document}


\title{\textsc{Mechanism Design for Demand Response Programs with financial and non-monetary (social) Incentives}}
\author{~Mateo~Alejandro~Cortés~Guzmán and Eduardo~Mojica-Nava}

\markboth{DRAFT, for submission to the IEEE Transactions on Smart Grids, \today}%
{Shell \MakeLowercase{\textit{et al.}}: Bare Demo of IEEEtran.cls for Journals}

\maketitle

\begin{abstract}
	
	Most demand management approaches with non-mandatory policies assume full users' cooperation, which may not be the case given users' beliefs, needs and preferences. In this paper we propose a mechanism for demand management including incentives both with and without money. The mechanism is validated by means of simulation, modeling the consumers as a finite multiagent system which evolves until a stable state, and social incentives diffusion using opinion dynamics.
	
\end{abstract}


\section{Introduction}

\IEEEPARstart{D}emand side management---DSM\cite{Samad2016}, understood with is broader connotation enclosing Demand Response, energy efficiency and another ways of modifying electricity demand --- is nowadays a widely discussed topic since optimization of the power system infrastructure and electricity use arose as topic of interest aiming to reduce climatic change and environmental impact of energy generation. In particular, the problem of demand response--DR-- has been approached from various perspectives as the technical side, involving direct load control, the financial side, with differential tariffs --day/night, hourly--, the informational side, with the inclusion of interactivity with the consumers through smart grids and from the game theory perspective, considering the problem as a Stackelberg game\cite{Stackelberg2012,Stackelberg2014,Stackelberg2016} where the service provider is a leader and the consumers are price takers, following leader's indication to maximize their profit.

Fotouhi Ghazvini \etal \cite{IncentivesFotouhiGhazvini2015} present an incentive-based demand response programs approach for retail electricity providers with DG--distributed generation assets and energy storage systems, starting from the fact that for some countries, there are financial risks regarding the price variations due to contracts made when buying energy for a price that is then sold for a smaller value. They formulate a two-stage scenario-based stochastic problem which aims to establish the required financial incentive and the optimal dispatch of DG units and energy storage systems usage--for the day-ahead market--, thus minimizing the ERP risk of financial losses by reducing power peaks--when electricity prices are higher-- Even if their rationale is not fully applicable in a country like Colombia\footnote{Where the electricity retailer has a guaranteed retribution for the energy it sells.} where the domestic users are regulated, the concept of including stochastic decision parameters can be adapted and applied to the particular case, adressing the uncertainty about consumers' preferences with the incorporation of stochastic decision parameters into their utility function.  They mention\cite{IncentivesFotouhiGhazvini2015} the difficulty of knowing electricity consumers' behavior and translating qualitative valuations to numeric variables that can be computed when designing DR programs. This paper addresses the problem by modeling the consumers as a multi-agent system --MAS--, and using opinion dynamics to model the social interaction and the evolution of users' preferences and beliefs, starting from an initial condition.



Thimmapuram \etal\cite{priceElasticityThimmapuram2013} show an analysis of the effect of price elasticity on deregulated markets, where they make assumptions about the consumers characterization and their responsiveness to electricity price variations and use a MAS approach\footnote{Using EMCAS software, a tool developed by the US Argonne National Laboratory for electric market related studies.} to model the adaptation process of the users. They state the inherent difficulty and complexity of measuring price elasticity, and estimations usually have high uncertainty. According to their recopilation of several elasticity studies, for residential costumers in the tradditional setting--without smart meters-- the price-load curve has a variable steep between $ -0.06 $ and $ -0.49 $ and in their simulation experiment, the achieved load reduction by changing price was about 5$\%$.


Regarding game theory and mechanism design approach, Nekouei \etal\cite{neko2015} have proposed a mechanism for demand management both at energy-market and small-consumers level. The former is treated as a Stackelberg game where generators are followers, and the latter is tackled with a VCG-like mechanism, ensuring consumers to reveal their true type--preferences--. The implementation is thought for application in South Australia, as for the power system structure and behavior. Haring \etal\cite{incentive_compatible2015} propose an incentive compatible mechanism as solution for the power supply in imbalance conditions requiring the use of ancillary services--AS--, treating the matter from the global perspective of the energy-market, without directly considering consumers' role but their aggregated demanded energy and power-flow in the transmission system.\footnote{110 kV and higher} The balancing energy parties--BRP--- make the resource allocation based on their willingness to pay, modeled as a polynomial function depending on the energy demanded from AS during a period. They maximize transmission system operator's utility, as well as BRP assets.

Mhanna \etal\cite{faithfulMechanismMhanna2016} propose a distributed mechanism for day-ahead residential demand scheduling, making an effort to get a realistic implementation of the designed mechanism. They address four common assumptions: first, the treatment of household energy levels as continuous\footnote{Continuous energy levels imply a convex optimization problem, hence easier to solve.}, when they are actually a mix of continuous and discrete levels. Mhana \etal~approximate the nonconvex problem by using mixed integer linear programming--MILP and an existing MILP solver.  Second, the representation of a household's--user-- valuation function by a concave increasing function\footnote{Following the diminishing marginal utility principle.}, which they assume rather constant and large for any device consumption profile--schedule-- that satisfies the user's preferences. In the present work, these assumption is adopted and modified, taking an evolving value for an user's valuation of her electricity consumption, modeled by means of opinion dynamics. Third, for the payment schemes based on an energy allocation that will be fully accurate, while it might not be the case; they propose a scoring rule-based payment which takes into account both day-ahead allocation and actual consumption, rewarding the accuracy of meeting the allocations. Finally, the implicit centralized implementation approach of mechanism design, where they take a distributed approach, in order to relax the computation requirements of a traditional centralized approach. With the previous assumptions, they get a mechanism that is asymptotically dominant incentive compatible, weakly budget balanced and fair in the charges to users--households--.


As electricity supply is in most places--if not all-- a public service and therefore must be granted to all inhabitants, enforcing policies for consumption change or mandatory taxes to higher consumptions are not really a practical option and there is no point in trying to reach an optimal use of the grid if consumers are not aware of it and willing to collaborate. Tackling this issue social networks models come to action, as consumers could be modeled as an information network of opinions reflecting users level of agreement to a certain measure or policy targeting DR. Furthermore, the DR program proposed to users can be not only the acceptance of financial incentives, but to include some policies with an environmental or social background, which requires users participation modifying their consumption according to a reference signal indicating if lower or higher consumption is required, and should be voluntarily followed appealing to moral conscience and empathy from people. 

By modeling the information layer of the users as a graph, opinion dynamics approach can be used to study users' beliefs evolution, constraining the demand response programs and incentives efficiency to a more realistic outcome considering that not all the society should actively participate in a certain DR program. A former approach of including social incentives for DR programs design is presented by Mojica-Nava \etal\cite{mojica2015opinion}, where users are assumed to follow a leader or prominent agent--which imposes the desired behavior-- and deviate from their natural preferences proportionally to the leader's influence exerted to each agent.

The foundation for including social incentives as modifiers for consumers behavior comes from its relationship with morality, discussed by several studies of altruist behavior in a society, giving some evidence that the actual behavior of people is not fully commanded by selfish principles, as stated by Adam Smith \cite{Gintis2005}, or as the evidence suggested by the recent study of Crockett \etal \cite{crockett2017moral} about moral behavior when we have a choice of whether harm or not to the fellow.

\section{Opinion Dynamics model for Domestic Electricity Consumers Characterization }

Graph theory has the tools to model the interactions between a group of agents whom may have some kind of relationship, and has been extensively used to model different kind of systems and networks, like internet, power systems--for a region or a country--, the network of scientific papers (citations) or interactions in social networks, just to name a few. As it would not be feasible to take the information of interactions and influences only from actual electricity consumers, the fact that many of the above mentioned real world networks can be modeled as a \textit{small world} network\footnote{An special type of random network.} is taken as base to model the interaction of users' opinion concerning the adoption or rejection of a DR oriented measure as a small-world random graph. Small-world networks are characterized for having a short average path between nodes and high clustering coefficient than classic random networks\cite{networks2011prettejohn}, while most nodes are connected to just a few neighbours, but also have a path to reach any other node in the network.

There are two well-known models for generating small-world networks that have been compared and validated with empirical data: the Watts-Strogatz--WS-- model, which reflects adequately the clustering seen in real networks, and the Barabási-Albert--BA-- model, in which there is a starting seed network which is growth using the \textit{preferential attachment} rule until the network desired size is reached. The BA model reflects better the degree distribution across the network, but does not represent clustering as good as WS model. In the application case presented in this paper, WS model is used.





If we represent the information exchange between a group of people with a graph, we can derive some analytical formulation of the group's interaction, namely \textit{opinion dynamics}, which aims to model the way interactions take place in a given group and how people's opinion about certain topic may evolve, and predict the actual choices made in the group.

One of the first opinion dynamics models is the one devised by DeGroot in 1974\cite{degroot1974consensus}, with his approach to reach a consensus of opinions between a group of persons considering the weighted adjacency matrix as an influence network. Friedkin and Johnsen in 1990\cite{friedkin1990social}, studied the topic as social influence network theory. Years later they came to propose a model\cite{friedkin1999social} of the interactions of a group of $n$ persons, taking into consideration the susceptibility to change of each individual, as well as the influence of every neighbor's opinion in an weighted average profit function that ultimately led to consensus. Their model is still the core for more recent social networks modeling, and unlike most other models, it has been validated in small and medium population groups for several problems\cite{tutorialSNProskurnikov2017}. Hegselmann and Krause\cite{hegselmann2002opinion} proposed a modified model taking into account a confidence factor, which limited the group of neighbors whose opinion influenced one person's particular beliefs, to those within a similarity range called \textit{bounded confidence}, maintaining the averaged opinion calculation with the neighbors in the bounded confidence subset. More recently, Mirtabatabaei \etal\cite{bullo_stubborn} have proposed a modification of the Friedkin-Johnsen's model, including coefficients to characterize various parameters of the social interaction of a group of individuals discussing a sequence of issues. The basic FJ dynamic as used by \cite{bullo_stubborn} is shown in eq.\ref{eq:bullo_stubborn},

\begin{equation}\label{eq:bullo_stubborn}
\begin{split}
\sigma_i(t + 1) = (1 - \mu_i)\sigma_i(0) + \mu_iA_{ii}\sigma_i(t) \\ + \mu_i(1 - A_{ii}) \sum_{j=1}^{n}A_{ij}\sigma_j(t) ;
\end{split}
\end{equation}

where, for the individual $i$, the term $\sigma_i(t)$ is the opinion at time $t$, $\mu_i$ the susceptibility to change her initial opinion, $A_{ii}$ the self-confidence of person about a particular issue, and $A_{ij}$ the credibility given to neighbors' opinion, contained in the  weighted adjacency matrix $A$ of thee graph, that is row stochastic with zero diagonal\cite{friedkin1999social}. \cite{bullo_stubborn} assumes some simplifications for Friedkin-Johnsen's model, regarding the coupling between prejudices and initial conditions. FJ model is proposed for both the consumer's valuation of energy and the willingness to enroll in a certain DR program that encourages users with alternate--non monetary-- benefits.

\section{MAS model of the Domestic Electricity Consumers} 

Taking a perspective from distributed control, electricity demand side can be modeled by means of evolutionary game theory\cite{mojica2015opinion,Barreto2014mechanism,handboookGTvol4}, where a group of individuals change their state over time trying to maximize their profit. 

The detailed evolution of a population of agents, \textit{i.e} how and when they update their strategies as time passes, is called a \textit{revision protocol}. The set of a \textit{revision protocol}, which models the population dynamic, and a \textit{population game}, which models the strategic environment for the players, define a stochastic evolutionary process\cite{sandholm2010population,handboookGTvol4}, that, according to the law of large numbers\footnote{Which states that given a large number of agents the noise due to extreme preferences is eliminated in the average behavior.}\cite{sandholm2010pairwise}, if the number of agents is very large it can derive a deterministic process called \textit{mean dynamic}.

According to Sandholm\cite{sandholm2010pairwise}, a \textit{revision protocol} $\rho^p$ is a map $\rho^p : \mathbb{R}^{n^p} \times X^p \to \mathbb{R}_+^{n^p \times n^p} $. The scalar $\rho_{ij}^p(\pi^p, x^p)$ is called conditional switch rate, from strategy $i \in S^p$ to strategy $j \in S^p$, given payoff vector $\pi^p$ and population state $x^p$. When using a revision protocol, it is assumed that each agent has a Poisson alarm clock with independent ring time, at rate $R$. For each alarm ring, the agent has an opportunity for strategy revision, switching to $j\neq i$ with probability $\rho_{ij}^p/R$ and keeping strategy $i$ with probability $1- \sum _{j\neq i} \rho_{ij}^p/R$. When an agent changes her strategy, population's state changes accordingly having into account the agent's decision.

One of the most studied and used is the \textit{replicator dynamics}, which is an imitation dynamic that conditions the change of strategy of a user to a repeated comparison with the profit of other population's individuals; if the individual $i$'s payoff is lower than that of individual $j$'s, she may update her strategy to the strategy played by $j$ and continue the interactions. According to Sandholm\cite{sandholm2010pairwise} pure replicator dynamics violate Nash stationarity, which is required for the dynamics to converge in a finite time, and it is necessary to include at least an small proportion of a mixed dynamic with direct strategy selection, to ensure that there are no permanently extinct strategies. The pairwise comparison revision protocol (eq. \ref{eq:parwise comparison}) is an example of direct strategy selection, and was first proposed by Smith, who derived from it the Smith dynamics. Candidate strategy $j$ is chosen randomly from $S$.

\begin{equation}\label{eq:parwise comparison}
	\rho_{ij}(\pi) =  [\pi_j - \pi_i]_+
\end{equation}

\begin{equation}\label{eq:parwise logit}
\rho_{ij}(\pi) =  \dfrac{exp(\eta^{-1}\pi_j)}{exp(\eta^{-1}\pi_j)+exp(\eta^{-1}\pi_i)}
\end{equation}

Another more complex example are the exponential revision protocols, like the \textit{pairwise logit} protocol, shown in eq.~\ref{eq:parwise logit}. For considerations of stochastic stability detailed in section~\ref{s: convergence}, this protocol is chosen in the MAS for modeling the strategy revision and update process that generates the dynamic in the game, as users will know available strategies set and can make a direct choice, rather than copying the behavior of other consumers. However, in the implementation phase it is made a modification to the protocol, regarding the random selection of candidate strategies. In order to accelerate convergence of the dynamic, the selection is weakly restricted to a strategy that enhances current payoff, similar to the best response dynamics\cite{sandholm2010population}.

In a natural non-incentivized state, the MAS model of the society receives as input the load profiles from the users reflecting their valuation of electricity use during each consumption period, which may be grouped under a few profiles --if it was desired to model groups with similar preferences-- or individual independent profiles for each consumer, and starting from an initial energy distribution, after the population dynamics evolution the model outputs a reconstruction of the load profile (see fig.~\ref{fig: power_allocation_noInc}).  With this setting, the natural state of the population can be altered through incentives.

\section{Mechanism design and incentives}

The demand response problem using mechanims design and population dynamics has already been adressed by other authors like Barreto \etal \cite{Barreto2014mechanism}, and the presented mechanism was at first based on their work. They propose an indirect revelation mechanism, based in VCG\footnote{Vickrey-Clarke-Grooves}\cite{handboookGTvol4}, incentive compatible and achieving pareto optimality if incentives are used during the transition period.

However, because of the difficulties that arose in the modeling of real-world users and their appliances, the mechanism implementation was changed by using a finite-population setting and revision protocols directly rather than performing analysis using the corresponding mean dynamic. The proposed mechanism is incentive compatible with indirect revelation.

\subsection{Game Setting}

Let a population $P$ be formed by a group of users--households--. Each user $i$ has a set of $N$ electricity powered devices, and some usage preferences $\theta_i\in \Theta $ for a number of consumption periods $T$, that can be--or not-- different for each user and each period. Preferences $\Theta$ are related to a set of strategies $S = {1,2,...\Sigma}$ for each consumption period, associated to the proportion of the period that an user prefers to use a device, \textit{i.e.} consume an specified amount of energy. Each user tries to maximize her utility, and the interaction within the whole society can be stated as a population game, where the fitness function for each user is her utility function, with dynamics led by the \textit{pairwise logit protocol}.

The utility function for each user $i$ can be described as her benefit of using the electricity minus its cost:

\begin{eqnarray}
\label{utilidad}
U_i = v_i - c_i\\
U_i = \sum_{n \in N}\alpha_i \vartheta_i^n - \beta(Q_{t}) q_i^n\\
U_P = \sum_{i \in P} U_i
\end{eqnarray}

Where $Q_{t} = \sum_{i \in P} q_i$ and $q_i$ is the energy consumption of user $i$ for a period, meaning that the cost for user $i$ depends on the electricity usage of the whole population. The electricity valuation $\alpha_i = f\beta\sigma_i$ where $f$ is an arbitrary factor. $\vartheta_i^n = \exp \left(- \dfrac{\lvert x_i - \theta_i\rvert}{\epsilon_i} \right)$ signals if the $i^{th}$ user consumption is according to her preference for device $n$, for a given period. Utility for each user-device combination has the form shown in fig.~\ref{f: utility_exp}, where the maximum surplus is obtained when the consumption matches her preference $\theta_i$ for the specific device. $\epsilon$ is an elasticity coefficient meaning how stiff is her consumption preference, and the total utility for each user is the sum of the utilities for the devices set.

\begin{figure}
	\begin{center}
		
		\includegraphics[width=0.45\textwidth]{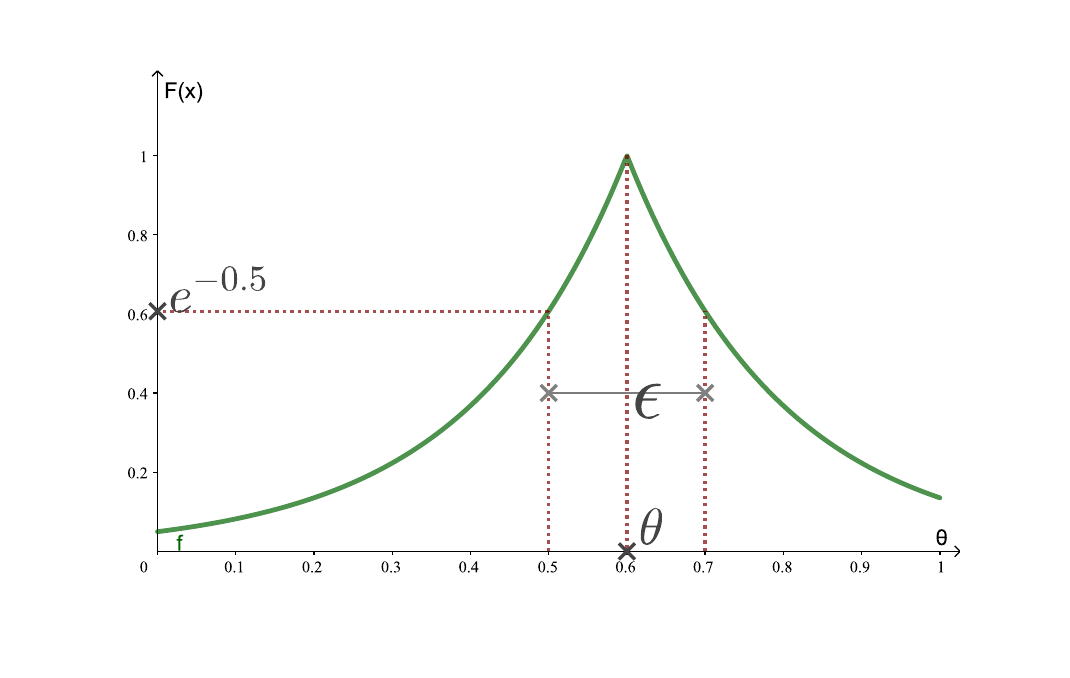}
		\caption{Example payoff range for the available strategies set.}
		\label{f: utility_exp}
	\end{center}
\end{figure}

The maximization problem solved for each period using and evolutionary approach can be stated as in eq.~\ref{eq:opt_problem}. As it is shown in eq. \ref{utilidad}, the users' profit depends of their accuracy following their own consumption preferences, and the surplus amount depends on the subjective valuation given by each user to the usage of the devices set they may own, with diverse functionality like TVs, AC units, PCs and so on. As users' preferences are private information, and even a given user may not be fully aware and conscious of her own, it is not feasible to know their true value and include it to the model as part of an utility function that reflects actual consumer behavior. It can be known, however, that in a population there will always be a wide range for these preferences where the minimum will be slightly higher than the consumed electricity price and the maximum will be many times the price paid for the electricity. Thus, population's valuation is a real vector $\alpha \in [c<\alpha_{min}, \alpha_{max}]$ shaped by the users' subjective opinion.

\begin{equation}\label{eq:opt_problem}
\begin{aligned}
& \underset{\bs{q}}{\text{maximize}}
& & U_P(\theta, q) =  \sum_{i \in P}\left( \sum_{n \in N}\alpha_i^n \vartheta_i^n - \beta(Q_{t}) q_i^n \right) \\
& \text{subject to}
& & q_i^n \geq 0,  i =\{1,\ldots,P\}, n =\{1,\ldots,N\}\\
& 
& & \Delta q_i^n \in \bf{Q_{devices}}\\
& 
& & \alpha_i^n \in [\beta, \alpha_{max}]\\
& 
& & \theta_i^n \in [0,1]
\end{aligned}
\end{equation}

In order to get an approximation to user preferences, it is assumed that preferences are embedded in the aggregated load profile for an user, and they can be extracted from it via virtual demand disaggregation, as in high-resolution load profiles generation. The following assumptions are considered:

\begin{assumption}
	\label{s preferencias}
	In natural state, consumption preferences for a given user form by aggregation her load profile.
\end{assumption}

\begin{assumption}
	\label{s valor_energia}
	Subjective users' valuation of energy consumption has an wide range that can be contained in a closed numeric interval with equivalent effect to the real-world diffuse valuation usually made by a person.
\end{assumption}

Assumption \ref{s valor_energia}, makes possible to propose an opinion dynamics for the evolution in time of the consumers' subjective valuation of energy, stating an starting condition (initial valuation), susceptibility and confidence--self and mutual-- parameters required by the Friedkin-Johnsen's model to compute final users' valuation state.

The population game can run in parallel for each time period, until the population state--load profile-- reaches a stable state.

\subsection{Monetary Incentives}

If we want to modify users' energy consumption, incentives must give users a higher surplus, for them to be willing to change their strategies shifting away from natural preferences. Punishment(negative) incentives are not considered in the present work, as forced changes are not desired unless there is a risk of outage or blackout. 

\begin{equation}
I \geq \alpha \vartheta - \beta(Q_{t}) q
\end{equation}

Monetary incentives can be fixed as:
\begin{equation}
	I = \gamma q_{redux}
\end{equation}

where $\gamma$ is the surplus for consuming away from preference and $q_{redux}$ the reduced energy consumption.$\gamma$ could be slightly larger than one, and its exact value when the mechanism is implemented should be derived empirically for each population under study.

Financial incentives should be constrained, as usually we don't want to reduce or increase consumption for all periods, but for those with certain characteristics, as the periods with power peaks. Depending on the specific interests of the central organizer, the periods qualifying for incentives can be set whether arbitrarily or as those exceeding certain consumption levels.

\subsection{Social (Non-monetary) Incentives}
It is likely that monetary incentives are not always worthy for domestic users considering their inherent low elasticity (Only those with very low valuation would be sensitive to changes in price). The other possibility to change energy use\footnote{By changing consumers' behavior. Another option behavior-independent is for example technology changes, making the same usage pattern (and users' utility) less energy intensive.} is by changing user's preferences directly, so they would have maximum utility by realizing the new preferences. If the initial belief(agreement) to a certain DR program, non-monetary motivated, are known, we can model its evolution using opinion dynamics, interpreting the result as a \emph{ probability of engage to the program} and shift preferences accordingly.

Social incentives are introduced in the system via opinion dynamics, associating $i^{th}$ user opinion to her probability to agree with the proposed demand response reference profile. This probability is also related to the subjective user valuation given to the alternate benefit associated with the DR program. Eqs.~\ref{eq:social incentives inicio} to \ref{eq:social incentives fin} show the incentivized utility function, where if the user decides to engage the DR program will adapt her preference until she reaches her flexibility limit, $\varepsilon$.

\begin{eqnarray}\label{eq:social incentives inicio}
\bm C \sim Be(\bm \sigma^{dr})\\
\bm \theta_{inc} = \bm \theta \pm \bm C \bm \varepsilon; \theta_{inc} \in [0,1]\\
\bm U_{i.soc} = \bm \alpha \exp \left(- \dfrac{\lvert \Delta \bm\theta_{inc}\rvert}{\epsilon}\right) -\bm \beta \bm q 
\label{eq:social incentives fin}
\end{eqnarray}


Where $\sigma^{dr}$: probability to engage with DR program. Thus, the social incentive is the difference between the natural and incentivized utility:

\begin{equation}
\bm I_{soc} = \bm U_{i.soc} - \bm U
\end{equation}

\subsection{Mechanism Properties}

Understanding the mechanism as the allocation rule and the resulting social choice, the proposed mechanism has the following properties:
\begin{itemize}
	\item Incentive compatibility (strategy proofness). As it can be seen in fig.~\ref{f: utility_exp}, the exponential utility function has an unique value peak for each given user and device preference. Thus, the dominant strategy is the one maximizing utility. 
	\item Indirect revelation, as preferences remain private and only aggregate consumption is spread.
	\item Contrary to traditional approach, social choice does not need to choose from society individual preferences, but states the energy consumption associated to the aggregated preferences.
\end{itemize}

\subsection{Convergence}\label{s: convergence}

Convergence and stability of the proposed model depends on convergence of the forming dynamics, \ie~for opinion and population. Population dynamics using pairwise comparison revision protocol satisfy both Nash Stationarity--N.S. and Positive Correlation--P.C., which are conditions that ensure long term stability for the mean dynamic. Common studied evolutionary games usually assume an infinite-population setting, which allows to average away the stochastic noise in the agent's evolution, thus allowing a deterministic behavior explained as the mean dynamic. In the finite-population setting, the agents' evolution can be seen as a Markov process $ \{X^N_t\} $ with common jump rate between states, and transition probabilities related to the payoffs of the revision protocol in use.

Finite-state Markov processes are characterized by their stationary distribution\footnote{If a process is described by this distribution at one time, it will be described by the same distribution in any other time \cite{sandholm2010population} } in very long time spans. Particularly, if the process has the property of \textit{irreducibility}, which means that the process has non-zero probability for visiting any state and will visit all possible states in finite time, then the stationary distribution is unique and does not depend on initial conditions; if it fulfills the property of \textit{reversibility}, then the stationary distribution adopts a simpler form and its calculation becomes feasible for practical cases. Only two-strategy games (under an arbitrary revision protocol) and potential games under exponential revision protocols are known to be reversible\cite{sandholm2010population, Sandholm2015, Wallace2015}.

For the present	demand management problem we can fix the revision protocol as \textit{logit} or \textit{logit pairwise} as the game always has more than two strategies but when only binary consumption levels are considered, thus meeting the first condition for known reversibility. The other condition needs to demonstrate that the game proposed is a potential game. 

When the number of agents is small, usually the effect of a particular agent's change of strategy in the population state is not negligible; this is considered in the \textit{clever payoff}\cite{sandholm2010population} evaluation, which has into account that the agent makes her comparison against the modified population.

\begin{proposition}
	The finite-population game with fitness function $$F(\Theta) = \exp \left(- \dfrac{\lvert \Delta\theta\rvert}{\epsilon} \right) -\beta q$$ where $\theta, \epsilon, q \in \Theta$ is the preferences matrix for the whole population, is a potential game.
\end{proposition} 
 
\begin{proof}
	A finite-population game is a potential game if there is a finite-potential function $f^N$ such that for any agent $i$ the payoff is defined by its discrete partial derivative\cite[\S 11.5]{sandholm2010population} : $	F^N_i(x) = f^N(x) - f^N(x - \dfrac{1}{N}x_i)$. If $\Delta\theta = x - \theta$, and $\Theta$ is fixed for a given user, we assume that $F(\Theta) = f^N$:
	\begin{eqnarray}\label{eq:stochastic stability}
	F^N_i(x) = \exp \left(- \dfrac{\lvert \bm x - \bm \theta \rvert}{\bm \epsilon}\right) -\bm \beta \bm q - \nonumber\\ 
	\left(\exp \left(- \dfrac{\lvert \bm x - \bm \theta \rvert}{\bm \epsilon}\right) -\bm \beta \bm q - \exp \left(- \dfrac{\lvert x_i^n- \theta_i^n \rvert}{\epsilon} \right)-\beta q_i^n\right)\\
	= \exp \left(- \dfrac{\lvert x_i^n- \theta_i^n \rvert}{\epsilon} \right)-\beta q_i^n\label{eq:stochastic stability fin}
	\end{eqnarray}
\end{proof}
As the proposed demand management game is a potential game with direct exponential (logit) revision protocol, then the underlying Markov process has stochastic stability, with stationary distribution\cite[\S 11.5,12.2]{sandholm2010population}:
\begin{equation}
\mu_x^N = \dfrac{1}{K^N}\dfrac{N!}{\prod\limits_{k \in S}(Nx_k)!} \exp (\eta^{-1}f^N(x)).
\end{equation}

$K^N$ is determined by the requirement that $\sum \limits_{x \in X^N}\mu_x^N = 1$.

\section{Mechanism Implementation, Results and Analysis}

\subsection{Users' profile generation}

The only available information regarding users' consumption in Colombia is the aggregated monthly consumption per user sorted by strata, and the consumption distribution across all strata registered by the network operator. In order to input the information required by the model, an algorithm for generating users' profiles from the aggregate consumption is needed. This algorithm should randomly generate a plausible profile, taking into account assumptions about users' electric devices and their consumption, usage patterns and aggregated monthly consumption so that the profile meets the monthly consumption for a real-world user.

González \etal\cite{Gonzalez2017,Moreno2017} propose a model where the load profile for an user is formed by the composition of several devices' profiles added up. Each device profile is generated as an stochastic process with some assumptions about typical usage for characterizing the probability distribution, and assumptions about nominal power, established in a way that the total energy for the generated profile is in the expected range of measured consumption.

Colombian population (including electricity consumers) is clasified among 6 strata (1 to 6) where users belonging to stratum 1, 2 and 3 have financial subsidies in the tariff, and users in stratum 5 and 6 are billed with over costs which are used to subsidize the former users. Additionally, a typical consumer belonging to a higher strata has higher income and increased electricity consumption, even though the higher percentiles of a low stratum overlaps with the lower percentiles of higher strata.

In order to reflect this clustering in the multiagent system, an algorithm for generating the variation was proposed. The load profiles are considered as Markov chains, as in \cite{markovLoadprofiles}, where they conclude that the method is useful for recreating overall energy consumption but with somewhat inaccurate reflecting exact behavior for individual devices. In the present case however, those inaccuracies are not of concern as the actual individual behavior is not known. Taking as starting point a sample profile generated by González \etal, which reflects approximately the behavior and consumption levels for stratum 4, the algorithm uses maximum likelihood estimation to get an approximate guess for the transmission and emission matrices of the underlying Markov process that generates the profile. The process is assumed to have two states for each device--on and off-- and the transition probabilities for related to the ON state are then changed multiple times in small steps until the mean consumption obtained with the modified Markov process reflects the consumption for a given stratum.

Then, with the transition and emission matrices for all strata available, the desired strata distributed profiles are generated as input for the evolutionary model.

\begin{algorithm}
	\caption{Markov process for strata estimation}\label{markovEstimation}
	\begin{algorithmic}[1]
		\Procedure{matrixAdjust}{}
		\State \emph{top}:
		\If {$t == 1$} \textit{TRANStry} $\gets$ \textit{TRANSref}
			\If {stratum > 4}
				\State \textit{increaseCons} = \textit{True}
				\EndIf
			\If {stratum < 4}
				\State \textit{increaseCons} = \textit{False}
			\EndIf
		\Else~{change \textit{TRANStry} by \textit{delta}}
		\EndIf
		\State
		\State from \textit{i}=1 to \textit{numTest}
		\State generate Markov chain with \textit{TRANStry}
		\State \textit{testEnergy(i)} $\gets$ energy consumption for case $i$
		\State 
		\State \textit{avgEnergy} $\gets$ mean(\textit{testEnergy})
		\If {abs(\textit{goalEnergy} $-$ \textit{avgEnergy}) $\leq tol$}
		\State \textit{TRANSadjusted} $\gets$ \textit{TRANStry}
		\State \textbf{return} \textit{TRANSadjusted} and \textbf{break}
		\Else 
		\State update \textit{increaseCons} depending on \textit{goalEnergy, avgEnergy}
		\State	$t \gets t + 1$
		\State \textbf{go to:} \textit{top}
		\EndIf
		\EndProcedure
	\end{algorithmic}
\end{algorithm}

\subsection{Population game}

PDToolbox\cite{toolbox} was used as starting point, making modifications in order to enable finite-multipopulation revision protocols and the required additions to include appropriately the utility and fitness functions, and opinion dynamics into the model. Modified code is available via GitHub at \url{https://github.com/macortesgu/PDToolbox_matlab}.

then fig.~\ref{fig: OpdynLowCollab} represents how individuals reach a final opinion about their collaboration in the DR program, being 0 none cooperation, and 1 full agreement to the requests.


Running the game for a week with 6-hours periods\footnote{For a real-world application case shorter periods are recommended (5-15 min), as some short duration details for the load profiles and preferences are lost when using long periods aggregation. However, for the sake of simplicity and clarity the presented application case uses long periods.} with $P = 40$ agents with average weekly consumption of approximately 70 kWh, and valuation factor $f = 7$\footnote{Means that the users value the electricity used maximum at 7 times its cost $\beta$.} the resulting power allocation for the natural case--without incentives-- is shown in fig.~\ref{fig: power_allocation_noInc}, where the close resemblance of the original aggregated load profile achieved by the MAS is evident, though they are not exactly the same, mainly because the stochastic nature of the finite-population dynamic. Next, for the same population several scenarios are tested, for both financial and social incentives.

\begin{figure}
	\begin{center}
		\includegraphics[width=0.45\textwidth]{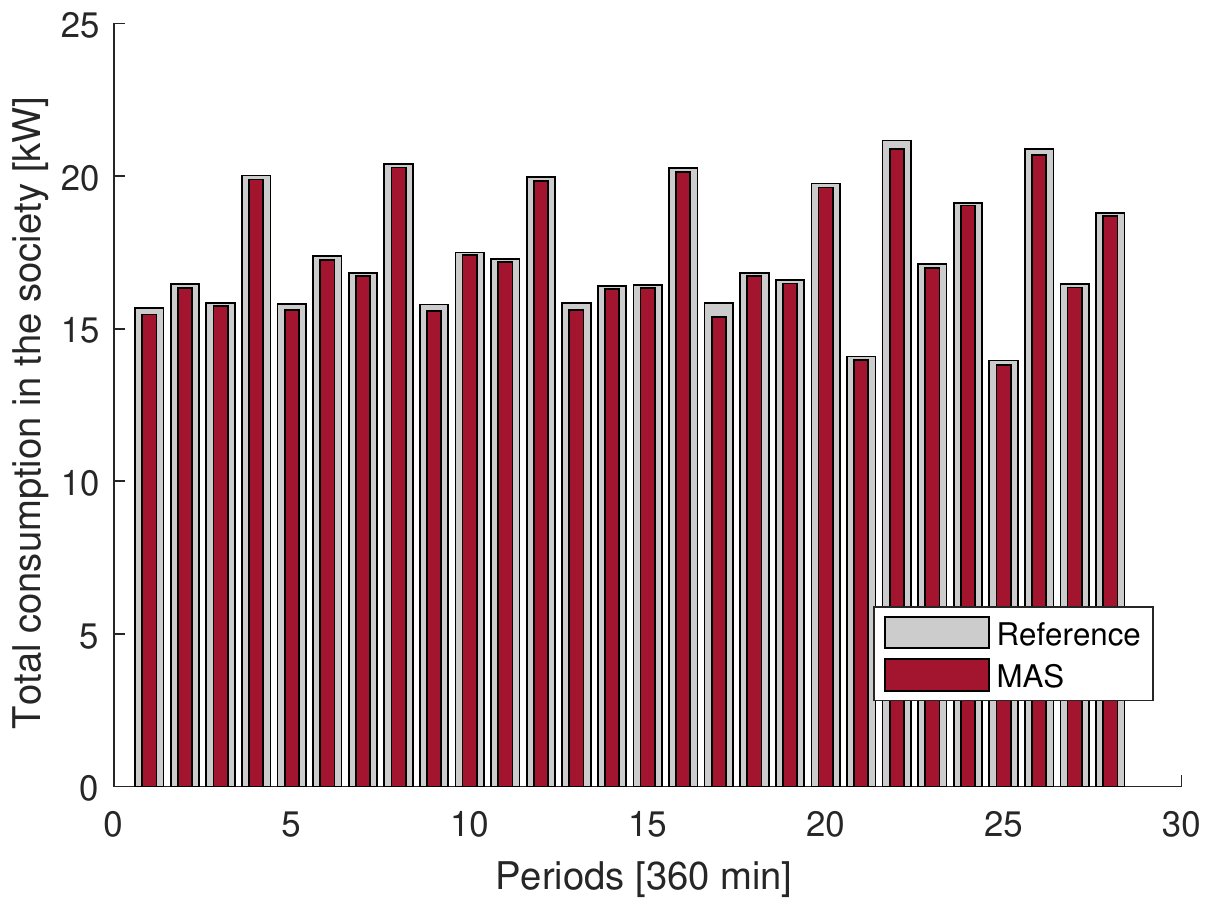}
		\caption{Power allocation for the 40 individuals population, modeled using Watts-Strogatz graph model.}
		\label{fig: power_allocation_noInc}
	\end{center}
\end{figure}

Social incentives are tested with flexibility $\varepsilon=0.3$ for a low-collaboration case, reflected in the willingness to adhere to the DR programs shown in fig.~\ref{fig: OpdynLowCollab}, a high-collaboration optimistic case related to opinions in fig.~\ref{fig: OpdynHighCollab} where almost everyone embraces the DR program, and a third case where users have high-collaboration and higher flexibility $\varepsilon=0.6$. The resulting power consumption evolution for the population is in fig.~\ref{fig: PowerEvoSocial} and the corresponding power allocation in figure~\ref{fig: PowerAllocSocial}. It can be seen the progressive power decrease when more users find the social incentive appealing and decide to enroll in the DR program, but also the longer times until stable state is reached. Due to the changes in utility, it is also more noticeable the stochastic nature of the MAS. For the high-collaboration case the power decrease is approximately $20\%$, and this could be considered as the maximum shift potential given the user's restrictions, as usually domestic consumers are highly inelastic and inflexible. In the figure~\ref{fig: PowerAllocSocial} is shown that incentives were applied in 11 of 28 periods, and the respective power during each period. The fitness evolution for one user in the case of low collaboration, disaggregated for each of her devices can be seen in figure~\ref{fig: FitnessEvo} where the step gains in fitness due to the incentive are clear for some of the devices.
  


\begin{figure}
	\begin{center}
		\includegraphics[width=0.35\textwidth]{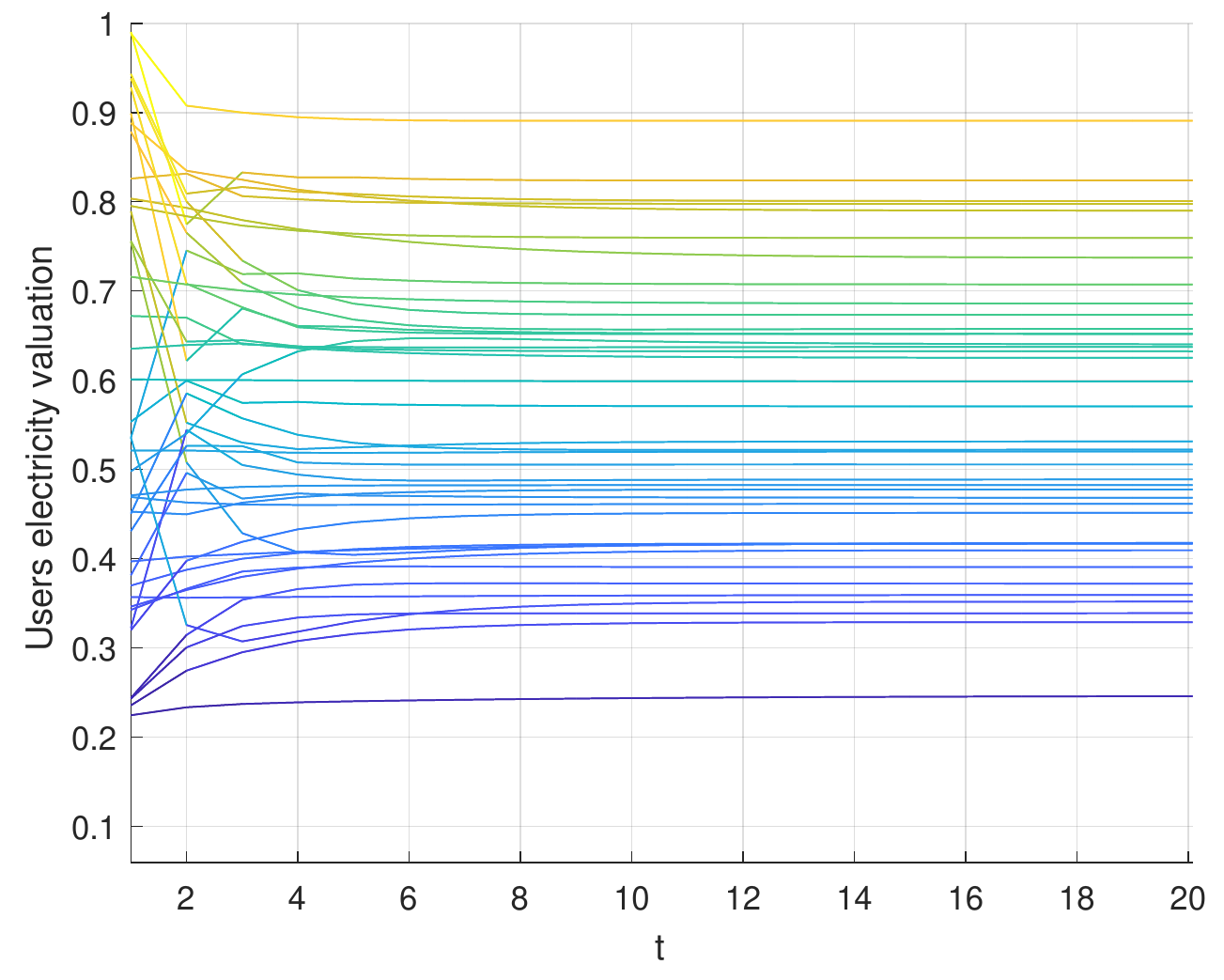}
		\caption{Users' electricity valuation (opinion dynamics), described by eq.~\ref{eq:bullo_stubborn}.}
		\label{fig: OpdynValuation}
	\end{center}
\end{figure}


\begin{figure}
	\begin{center}
		\includegraphics[width=0.35\textwidth]{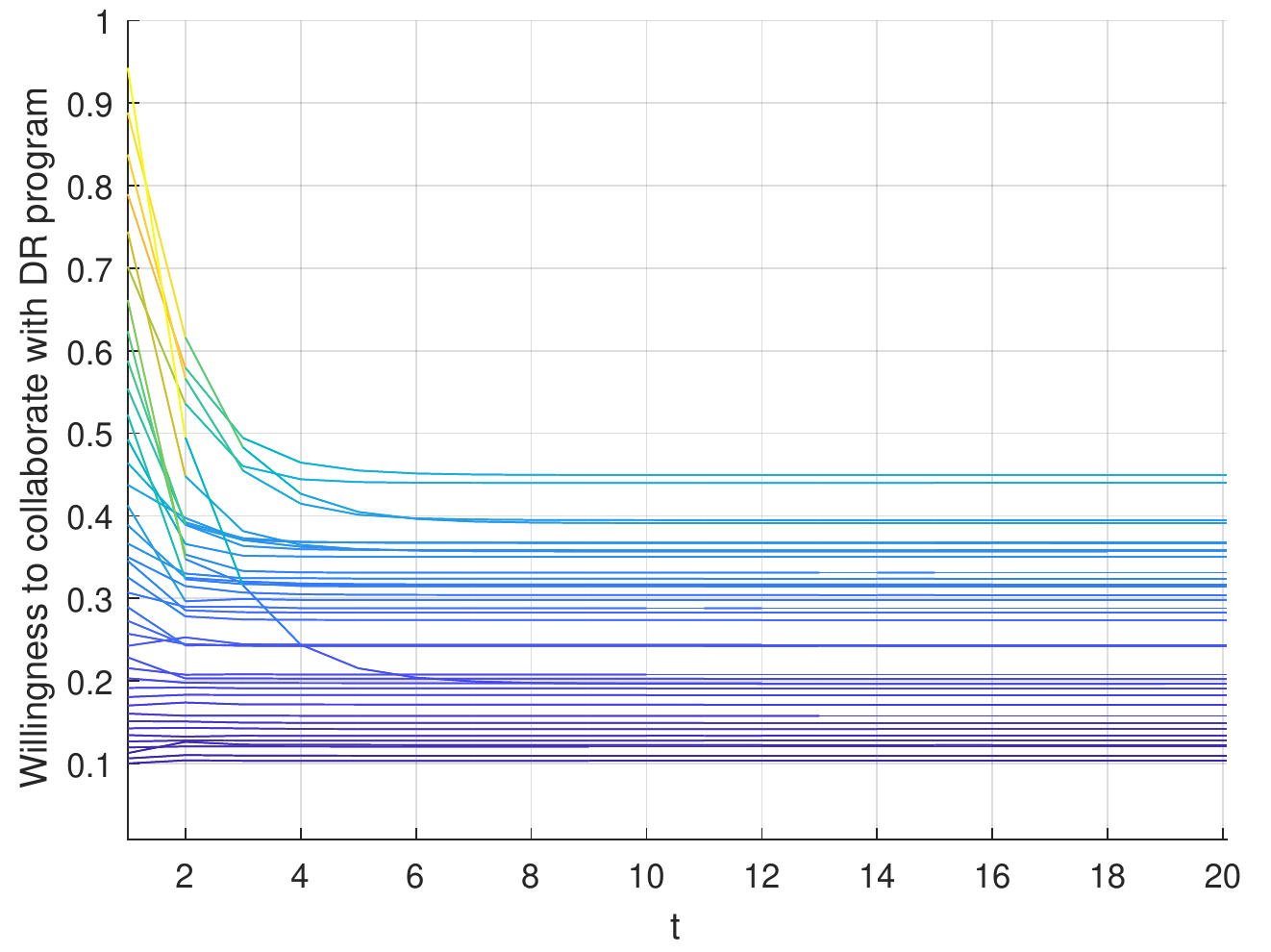}
		\caption{Users' collaboration, low collaboration scenario (opinion dynamics), described by eq.~\ref{eq:bullo_stubborn}.}
		\label{fig: OpdynLowCollab}
	\end{center}
\end{figure}

\begin{figure}
	\begin{center}
		\includegraphics[width=0.35\textwidth]{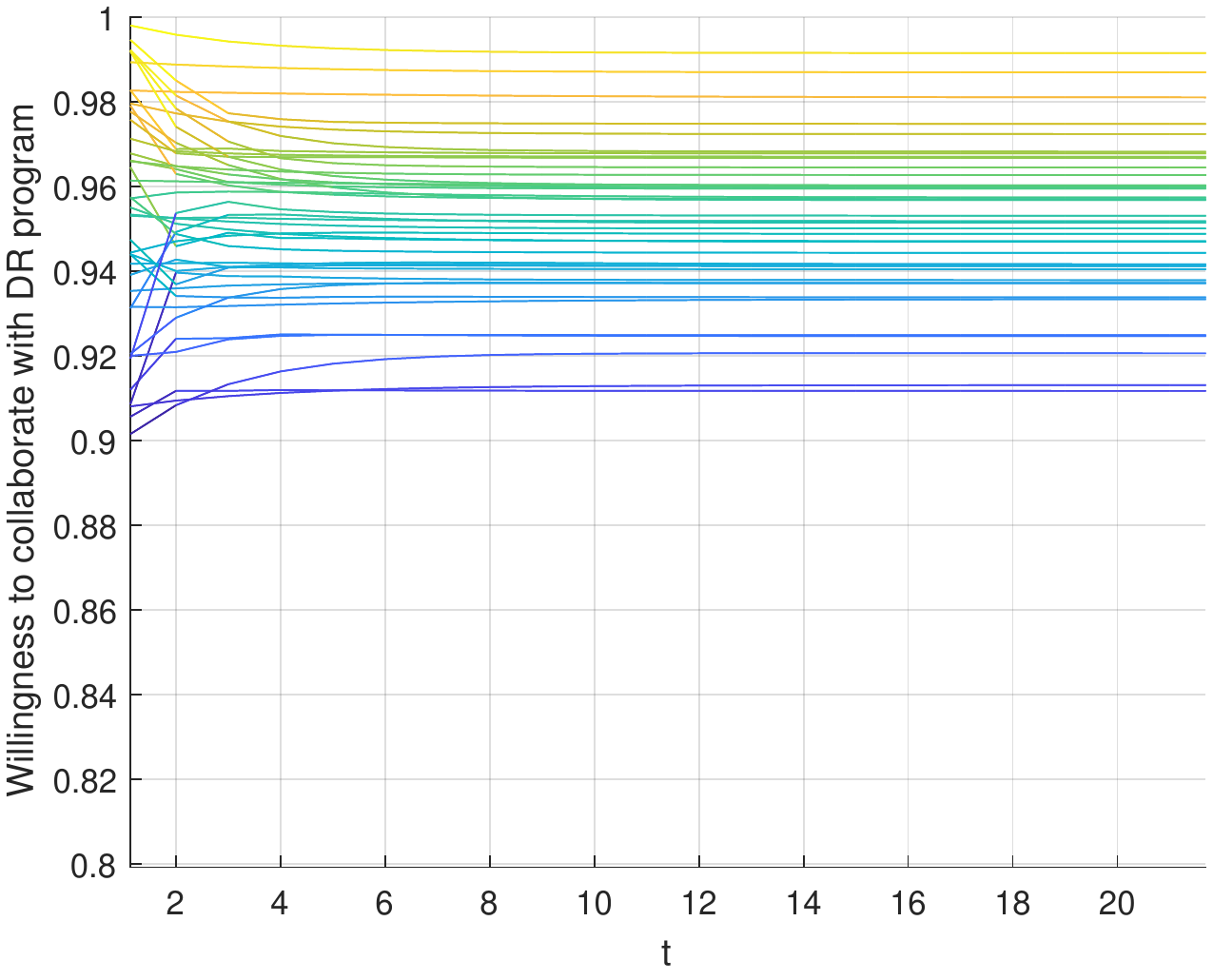}
		\caption{Users' collaboration, optimistic scenario (opinion dynamics), described by eq.~\ref{eq:bullo_stubborn}.}
		\label{fig: OpdynHighCollab}
	\end{center}
\end{figure}

\begin{figure}
	\begin{center}
		\includegraphics[width=0.45\textwidth]{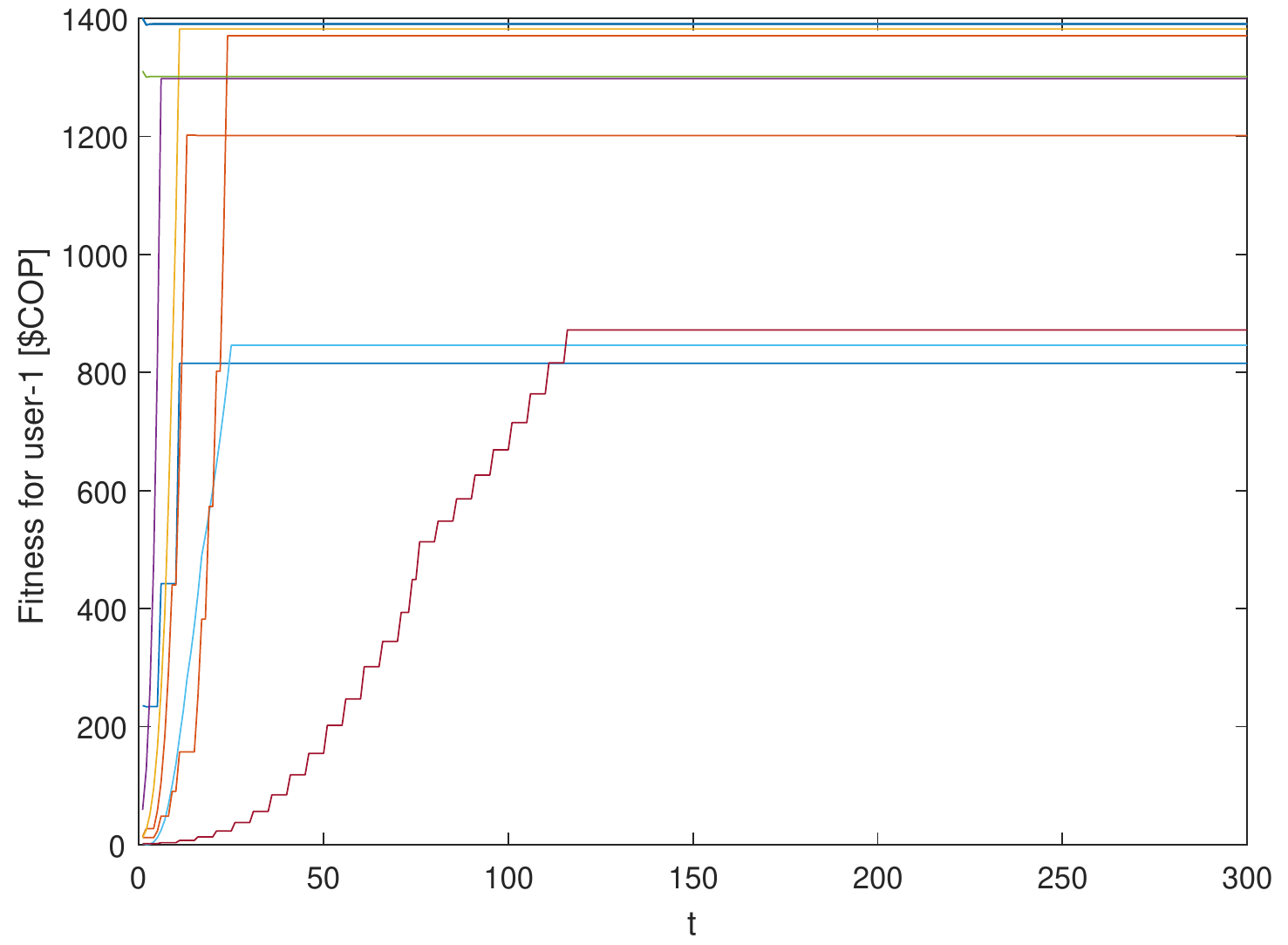}
		\caption{Fitness evolution for User 1 during the last period, low-collab scenario (opinion dynamics). Each line represents the fitness for each device the user has.}
		\label{fig: FitnessEvo}
	\end{center}
\end{figure}

\begin{figure}
	\begin{center}
		\includegraphics[width=0.45\textwidth]{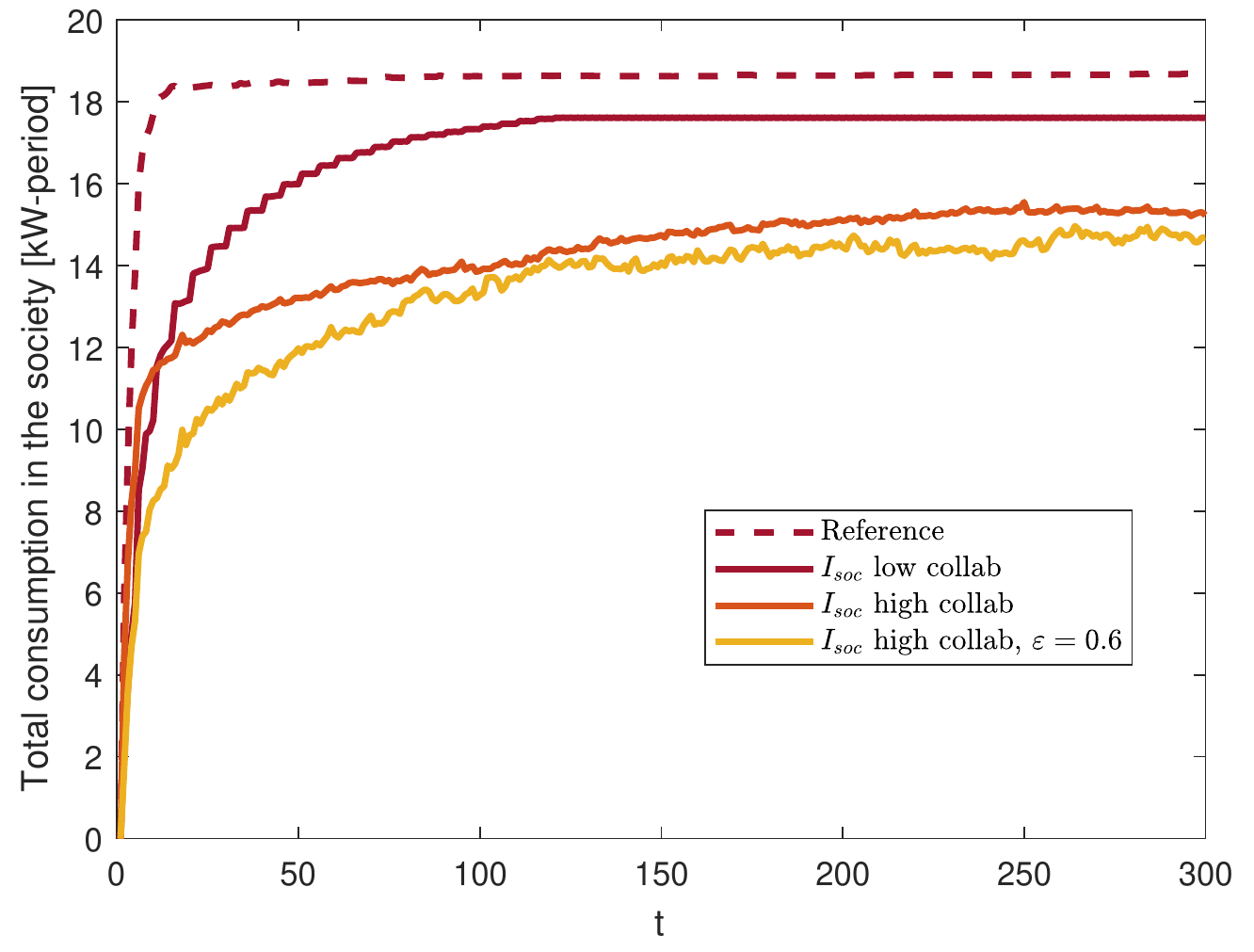}
		\caption{Power evolution for the 40 individuals population, for different incentives scenarios.}
		\label{fig: PowerEvoSocial}
	\end{center}
\end{figure}
\begin{figure}
	\begin{center}
		\includegraphics[width=0.48\textwidth]{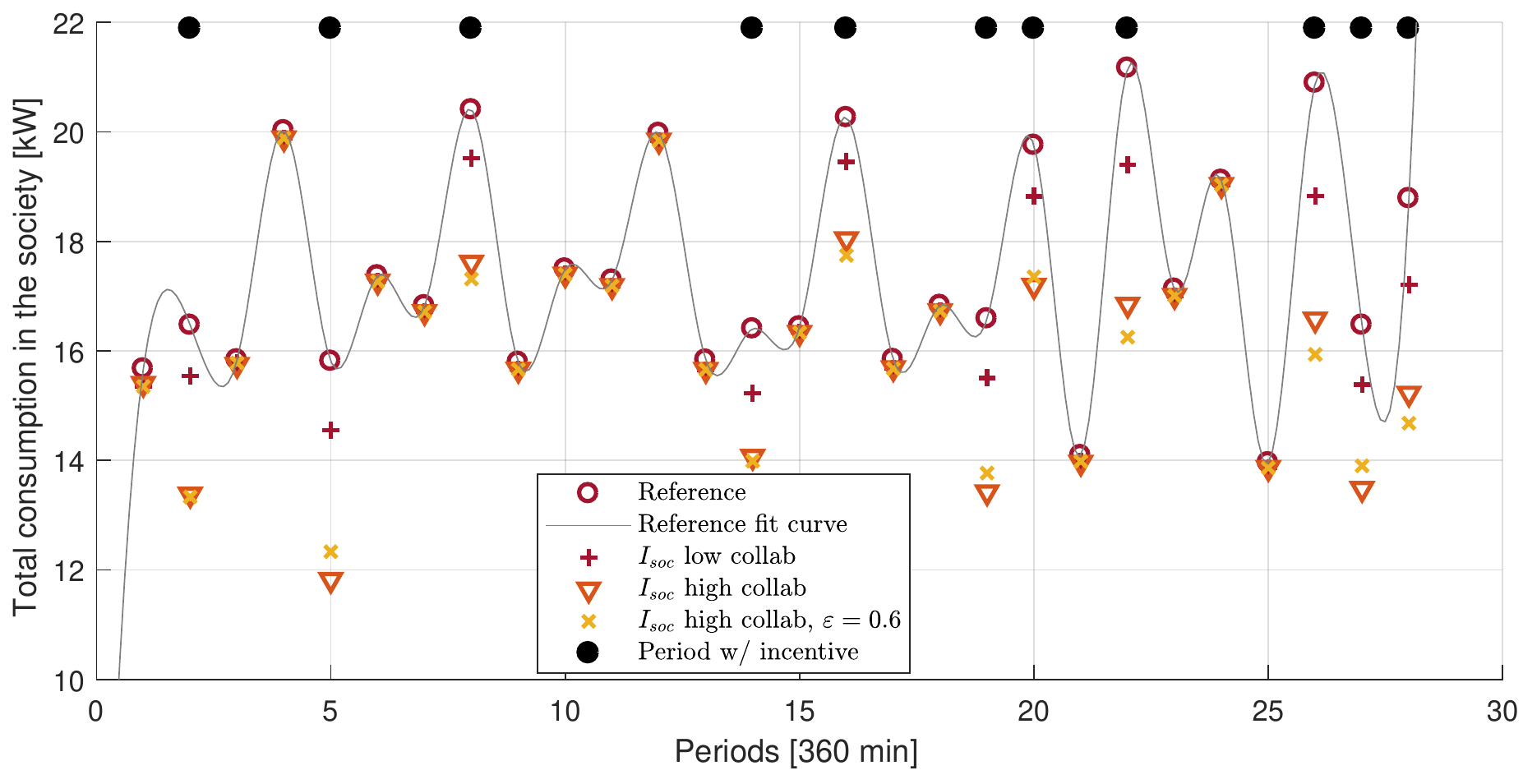}
		\caption{Final power allocation at each period for the 40 individuals population, for different incentives scenarios.}
		\label{fig: PowerAllocSocial}
	\end{center}
\end{figure}


Financial incentives are set for two cases with 26 of the 28 periods incentivized, when the compensation received by the users are $2\beta$ and $3\beta$ respectively. The corresponding power evolution for the population is shown in fig.~\ref{fig: PowerEvoFinancial} and the final power allocation per period in fig.~\ref{fig: PowerAllocFinancial}. The resulting reduction for the first case is roughly comparable to the one with low-collaboration social incentives, but the investment needed to fund the incentives can be substantially higher. The case where $I = 3\beta$ shows a noticeably higher decrease, where the collaborators are those whom value less their energy, and find in the incentive higher profit than following their consumption preferences. However, if the incentives need to be higher than the electricity cost in most cases it is no feasible that an interested agent as the network operator finds profitable the implementation of financial incentives, specially if her income is regulated and positive in the long term.

\begin{figure}
	\begin{center}
		\includegraphics[width=0.45\textwidth]{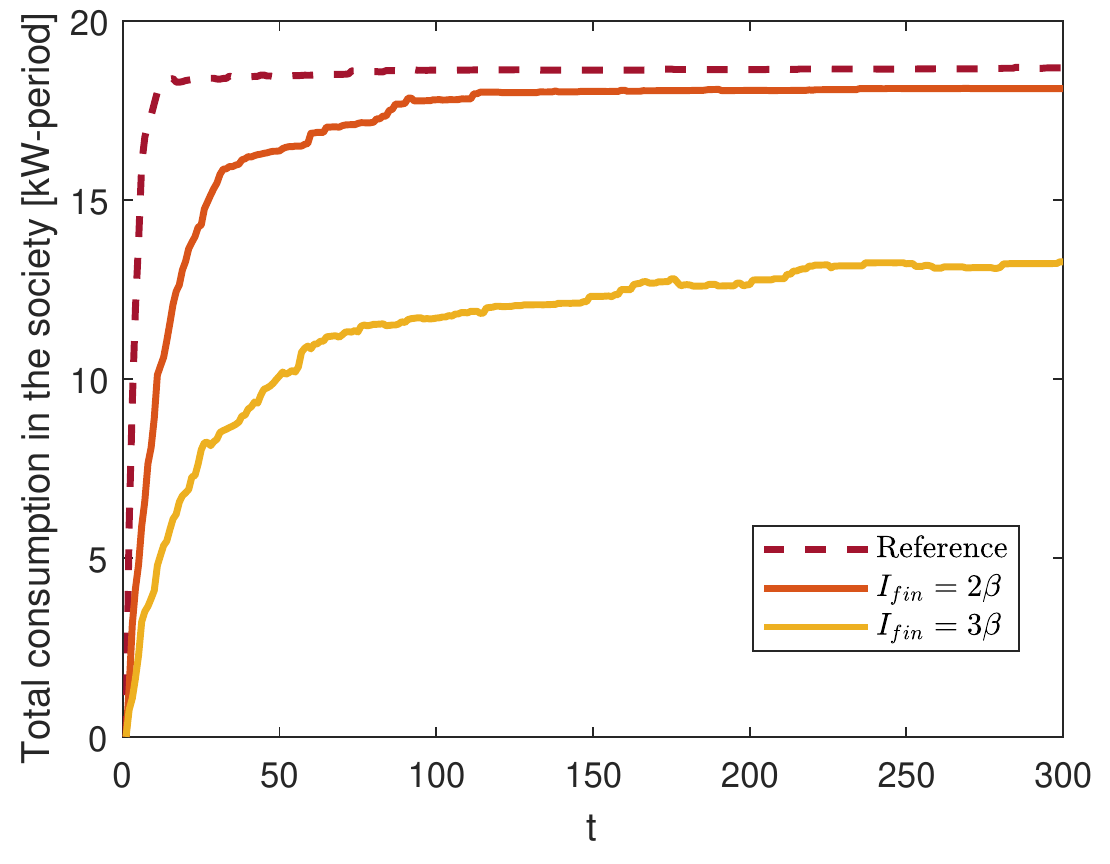}
		\caption{Power evolution for the 40 individuals population, for different financial incentives scenarios. $\beta$ is the electricity unit cost.}
		\label{fig: PowerEvoFinancial}
	\end{center}
\end{figure}
\begin{figure}
	\begin{center}
		\includegraphics[width=0.48\textwidth]{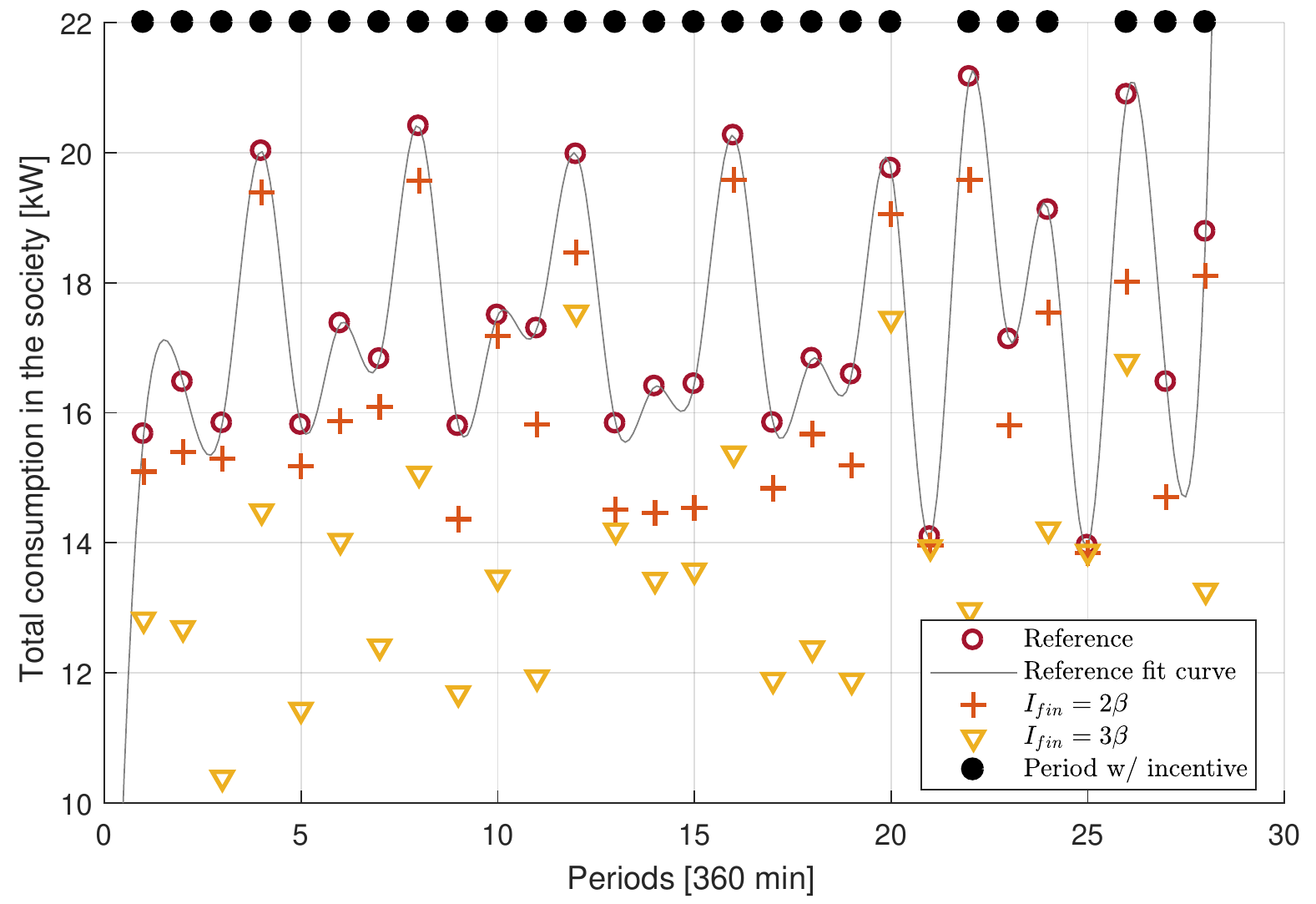}
		\caption{Final power allocation at each period for the 40 individuals population, for different financial incentives scenarios.}
		\label{fig: PowerAllocFinancial}
	\end{center}
\end{figure}

\section{Conclusions and Future Work}
Scenario-based predictions can be done by setting the opinion dynamics parameters with some desired values, and watching the effect this opinion has in the electricity consumption of users during a period, \textit{e.g.} a week.

An indirect revelation, incentive compatible mechanism for electricity demand management has been proposed. The mechanism maximizes domestic users' utility, achieved when their consumption matches their private preferences. Monetary and non-monetary incentives are also proposed for changing total consumed energy, whether it may be an increase or a decrease for certain consumption periods. In order to validate the mechanism, it has been proposed a MAS model of the users, that reproduce a given consumption profile in its natural state, before applying the desired incentives, helping to shift loads and lower peaks, within the technical restrictions of the users. A test scenario where social incentives are used to increase consumption efficiently constrained by users' purchasing power is studied by Cortés \etal in \cite{cortesDR2017}.


The proposed MAS is based on a finite population game and uses discrete energy levels that reflect the actual characteristics of commonly used home appliances and devices. It also has a fitness function that accounts for users subjective valuation of electricity, given the specific device that uses the energy, obtaining a more realistic approach to the problem.

The precise effect of the network characteristics, type of random network used in the model and stubborn or prominent agents has yet to be determined, as well as further adjustments contrasting the predictions obtained with the simulations with real-world scale deployment of incentives.

\section{Acknowledgments}

Partial financial support by Colciencias granted to one of the authors in the call No. 706-2015 “Convocatoria Nacional Jóvenes Investigadores e Innovadores año 2015” is gratefully acknowledged.



\bibliographystyle{plain}
\bibliography{../bib-file/Fullbiblio}

\end{document}